\documentclass[a4paper,preprint,11pt]{elsarticle}
\usepackage{amsmath, amssymb, amsfonts, amsthm, float, hyperref}
\usepackage{enumitem}
\usepackage{placeins} 
\usepackage[caption=false]{subfig} 
\DeclareMathOperator\erf{erf}

\newtheorem{theorem}{\bf Theorem}[section]
\newtheorem{corollary}{\bf Corollary}[section]
\newtheorem{definition}{\bf Definition}[section]
\newtheorem{remark}{\bf Remark}[section]

\begin{document}
\begin{frontmatter}
\title{Stability of calibration procedures: \linebreak fractals in the Black-Scholes model}

\author[add1]{Yiran Cui\corref{cor}}
\ead{y.cui.12@ucl.ac.uk}
\author[add2]{Sebastian del Ba\~{n}o Rollin}
\ead{s.delbanorollin@qmul.ac.uk}
\author[add1,add3]{Guido Germano\corref{funding}}
\ead{g.germano@ucl.ac.uk, g.germano@lse.ac.uk
}

\address[add1]{Financial Computing and Analytics Group, Department of Computer Science, \linebreak University College London, United Kingdom}
\address[add2]{School of Mathematical Science, Queen Mary University of London, United Kingdom}
\address[add3]{Systemic Risk Centre, London School of Economics and Political Science, \linebreak United Kingdom}

\cortext[cor]{Corresponding author. }
\cortext[funding]{This author acknowledge the support of Economic and Social Research Council (ESRC) in funding the
Systemic Risk Centre (grant number ES/K002309/1).} 

\begin{abstract}
Usually, in the Black-Scholes pricing theory the volatility is a positive real parameter. 
Here we explore what happens if it is allowed to be a complex number.
The function for pricing a European option with a complex volatility has essential singularities at zero and infinity.
The singularity at zero reflects the put-call parity.
Solving for the implied volatility that reproduces a given market price yields not only a real root, but also infinitely many complex roots in a neighbourhood of the origin.
The Newton-Raphson calculation of the complex implied volatility has a chaotic nature described by fractals.
\end{abstract}

\end{frontmatter}

\section{Introduction}\label{sec:fractal_intro}
Every day, the implied volatility consistent with the given price of a European option is computed millions of times in trading and risk management systems throughout the financial industry. 
This is typically done with the Newton-Raphson method, which can exhibit chaotic phenomena when hunting a successively better approximation to the root. 
These phenomena are best described in the complex plane by means of the associated fractal Julia sets. 
In a one-page paper published in 1879, Cayley \cite{cayley1879} first suggested the difficulty in extending the Newton-Raphson method (which he called Newton-Fourier) to the following cases:
\begin{enumerate}[series=CayleyClass]
\item complex polynomials with degree higher than 2;
\item a complex initial value leading the following iterations;
\item the allowance of complex roots.
\end{enumerate}
Furthermore, in the study of this problem he proposed a concept which later was called an attraction basin:
\begin{quotation}
\emph{``To determine the region of the plane, such that $P$ (the initial point) being taken at pleasure anywhere within one region we arrive ultimately at the point $A$ (a root); anywhere within another region at the point $B$; and so for the several points representing the roots of the equation.''}
\end{quotation}

In this paper, we extend the Black-Scholes valuation \cite{Black1973} to a complex implied volatility parameter, allowing the initial value of the Newton-Raphson method to be complex; then we explore the fractal objects that describe the chaotic nature of the Newton-Raphson calculation of the implied volatility.

\subsection{Implied volatility}
The Black-Scholes model does not adequately take into account important characteristics of the market dynamics such as skewness, fat tails and the correlation between the asset's value and its volatility. Other models have been devised to better approximate the fair price of derivatives, as discussed in a large body of research.
However, dealers still prefer to describe the price of an option $V$, obtained either by these refined models or from a market quote, in terms of the volatility $\sigma$ such that the Black-Scholes formula replicates the given price. 
This parameter $\sigma$ is called the \emph{implied volatility} and is often described, following Rebonato \cite[pg.~169]{rebonato2004}, as:
\begin{quotation}
 \emph{``the wrong number to put in the wrong formula to get
the right price of plain-vanilla (European) options''}.
\end{quotation}
From the perspective of a trader, implied volatility results from a rescaling process that allows to compare the relative worth of options with different maturities or involving different assets or currencies, where a crude comparison in terms of premium would be inapplicable. For similar reasons it is also used in the interpolation of prices of options with different maturities and strikes.

\subsection{Numerical scheme to calculate the implied volatility}
To calibrate the implied volatility $\sigma \in \mathbb{R}_+$, the Newton-Raphson method is used to solve the equation which matches the market price $V$ and the Black-Scholes valuation for a European option:
\begin{equation} \label{Black-Scholespv}
V=e^{-rT}\theta
\left\{
F\Phi\left[\theta\left(\frac{\log\frac{F}{K}}{\sigma\sqrt{T}}+\frac{1}{2}\sigma\sqrt{T}\right)\right]
-
K\Phi\left[\theta\left(\frac{\log\frac{F}{K}}{\sigma\sqrt{T}}-\frac{1}{2}\sigma\sqrt{T}\right)\right]
\right\},
\end{equation}
where $\theta = 1$ for call options and $\theta = -1$ for put options,
$F:=S_0e^{(r-d)T}$ is the fair forward price to maturity $T$, $r$ is the (domestic) risk-free interest rate, $d$ is the dividend rate (or foreign interest rate for a foreign exchange rate contract), $S_0$ is the spot price of the underlying, $K$ is the strike price, and $\Phi(\cdot)$ is the standard normal cumulative distribution function
\begin{equation}
\Phi(x)=\frac{1}{\sqrt{2\pi}} \int_{-\infty}^x \exp\left(-\frac{u^2}{2}\right)\mathrm{d}u,
\end{equation}
which can be expressed through the error function as
\begin{subequations}
\begin{align}
\Phi(x)
& =\frac{1}{2}+\frac{1}{\sqrt{2\pi}}\int_{0}^{x} \exp\left(-\frac{u^2}{2}\right)\mathrm{d}u\label{eq:RealNormalCDF}\\
&= \frac{1}{2}+\frac{1}{2}\erf \frac{x}{\sqrt{2}}.
\end{align}
\end{subequations}

If we denote by $f(\sigma)$ the right-hand side of equation \eqref{Black-Scholespv}, then the Newton-Raphson iteration to solve $V = f(\sigma)$ is given by
\begin{equation}\label{Black-Scholesiteration}
\sigma_{n+1}=\sigma_n -\frac{f(\sigma_n)-V}{f'(\sigma_n)},
\end{equation}
starting from an arbitrary initial guess $\sigma_0\in\mathbb{R}_+$.
However, as explained by J\"ackel \cite{jackel2006}, $f(\sigma)$ is convex for low volatilities and concave for higher volatilities, causing instabilities in the algorithm. J\"ackel's article demonstrates how taking the logarithm of both the market price and the Black-Scholes price overcomes this convergence problem.

Inspired by the chaotic phenomena arising in the Newton-Raphson search for an implied volatility close to the origin, where $f(\sigma)$ is too flat according to J\"ackel, we perform a new experiment on the calibration of implied volatility where its search domain is extended to the complex plane by starting from an initial guess $\sigma_0 \in\mathbb{C}$ and the Black-Scholes price is extended as an analytic function on $\mathbb{C}_*=\mathbb{C}\setminus\{ 0\}$ with essential singularities at zero and infinity. We observe infinitely many complex roots for equation \eqref{Black-Scholespv} other than the real one, as will be illustrated by means of fractal attraction basins.

\section{Analytic extension of the pricing function}
In this section, we show the analytic extension of the major component of the Black-Scholes pricing function, i.e., the normal cumulative distribution function, and therefore the extension of the pricing function itself. We discuss the singularity of this function and hence the complex roots of equation \eqref{Black-Scholespv}. Illustrations are given at the end.

\subsection{Analytic extension of the cumulative normal distribution}
The standard normal cumulative distribution function $\Phi(z)$ with a complex argument $z\in\mathbb{C}$ is related to several special functions that arise often in applied mathematics and engineering;
see for example Fettis, Caslin and Cramer \cite{fettis1973} for an analysis of the zeros of $\erf(z)$.% (the complementary error function alluded to in \emph{loc.cit.} is $\Phi(-z)/2$).
\begin{theorem}
The cumulative density function of the standard normal distribution can be extended as a complex entire function on $\mathbb{C}$.
\end{theorem}

\begin{proof}
Given that $\exp(-z^2/2)$ is an entire function, by general results of complex analysis \cite{ahlfors1966} the function 
\begin{equation}
\Phi(z)=\frac{1}{2}+\frac{1}{\sqrt{2\pi}}\int _{0}^{z} \exp\left(-\frac{u^2}{2}\right)\mathrm{d}u
\end{equation}
is well defined and entire. Trivially, its restriction to the real line is the function $\Phi(x)$ defined in equation \eqref{eq:RealNormalCDF}.
\end{proof}

An analytic function is said to have an isolated singularity at a point if the function is analytic in a neighbourhood of the point with the point excluded. 
Isolated singularities of analytic functions in one variable are classified as \cite{ahlfors1966}:
\begin{description}
\item[Removable] 
if the function can be assigned a value at that point such that the resulting extended function is analytic. A typical example of a removable singularity is $z=0$ for $f(z) = (\mathrm{sin}z)/z$.
\item[Pole] 
if the norm of the function tends to $\infty$ as that point is approached. A typical example of a pole singularity is $z=0$ for $f(z)=1/z$. %with $n>0$.
\item[Essential] 
in all other cases. A typical example of an essential singularity is $z=0$ for $f(z)=\exp(1/z)$.
\end{description}

\begin{remark}\label{rmk:CDFalongIm}
The complex cumulative normal distribution function $\Phi(z)$ has an essential singularity at $z=\infty$. 
This is because along the real axis, when $z\to +\infty$, $\Phi(z)\to 1$ and when $z\to -\infty$, $\Phi(z)\to 0$. 
On the other hand, along the imaginary axis, $\Phi(z)$ is unbounded for $z \to \pm i\infty $: for real $y$, substituting $u = iv$,
\begin{subequations}
\begin{align}
\Phi(iy)&=
\frac{1}{2}
+
\frac{1}{\sqrt{2\pi}}\int _{0}^{iy} \exp\left(-\frac{u^2}{2}\right)\mathrm{d}u \\ 
&=
\frac{1}{2}
+i
\frac{1}{\sqrt{2\pi}}\int _{0}^y \exp\frac{v^2}{2} \mathrm{d}v,
\end{align}
\end{subequations}
which grows to $\pm i\infty$ as $y\to\pm\infty$.
\end{remark}

\subsection{Analytic extension of the Black-Scholes pricing formula}

\begin{theorem}
The Black-Scholes price as a function of the volatility $f(\sigma)$ can be extended as an analytic function on  $\mathbb{C}_*=\mathbb{C}\setminus\{ 0\}$. The singularities at zero and infinity are essential.
\end{theorem}

\begin{proof}
The necessary and sufficient condition for a point $z_0$ to be an essential singularity of $f(z)$ is that $\lim_{z\to z_0} f(z)$ does not exist. 
Let $d_\pm:=\log(F/K)/(\sigma\sqrt{T})\pm\sigma\sqrt{T}/2$, then it is easily seen that when $\sigma$ approaches zero along the imaginary axis, $d_\pm$ approach $+i\infty$ or $-i\infty$, depending on the ratio $F/K$. As shown in remark \ref{rmk:CDFalongIm}, $\lim_{y\to \infty}\Phi(iy)$ is indefinite. Thus $\lim_{z\to z_0} f(z)$ equals neither a finite complex number nor $\infty$, i.e.,\ the limit does not exist. Similarly one can verify that the singularity at $\sigma=\infty$ is essential.

From another perspective, the singularity is not removable, as otherwise it would have a zero Taylor expansion, and it is not a pole either as otherwise the function would tend to infinity.
\end{proof}

\begin{remark}\label{rmk:fractal_putcall}
If we denote by $f(\sigma)$ the RHS of equation \eqref{Black-Scholespv}, then we have that $f(-\sigma)$ is the opposite of the price of the put option with the same maturity and strike. This has as a consequence that along the real axis
\begin{eqnarray}
\lim _{\sigma\rightarrow 0^+} f(\sigma) &=& 
\begin{cases} e^{-rT} (F-K) &\mbox{if } F\geq K \\
0 & \mbox{if } F\leq K \end{cases} 
\\
\lim _{\sigma\rightarrow 0^-} f(\sigma) &=& 
\begin{cases}  0&\mbox{if } F\geq K \\
e^{-rT} (F-K) & \mbox{if } F\leq K, \end{cases} 
\end{eqnarray} 
which is reflected in figure \ref{fig:modulusBlack-Scholes} 
and again implies that the singularity at $\sigma=0$ is essential. 
By the put-call parity $f(\sigma)-f(-\sigma)$ is independent of $\sigma$.
\end{remark}

\begin{figure}[!bt]
\centering
 \includegraphics[width=.6\textwidth]{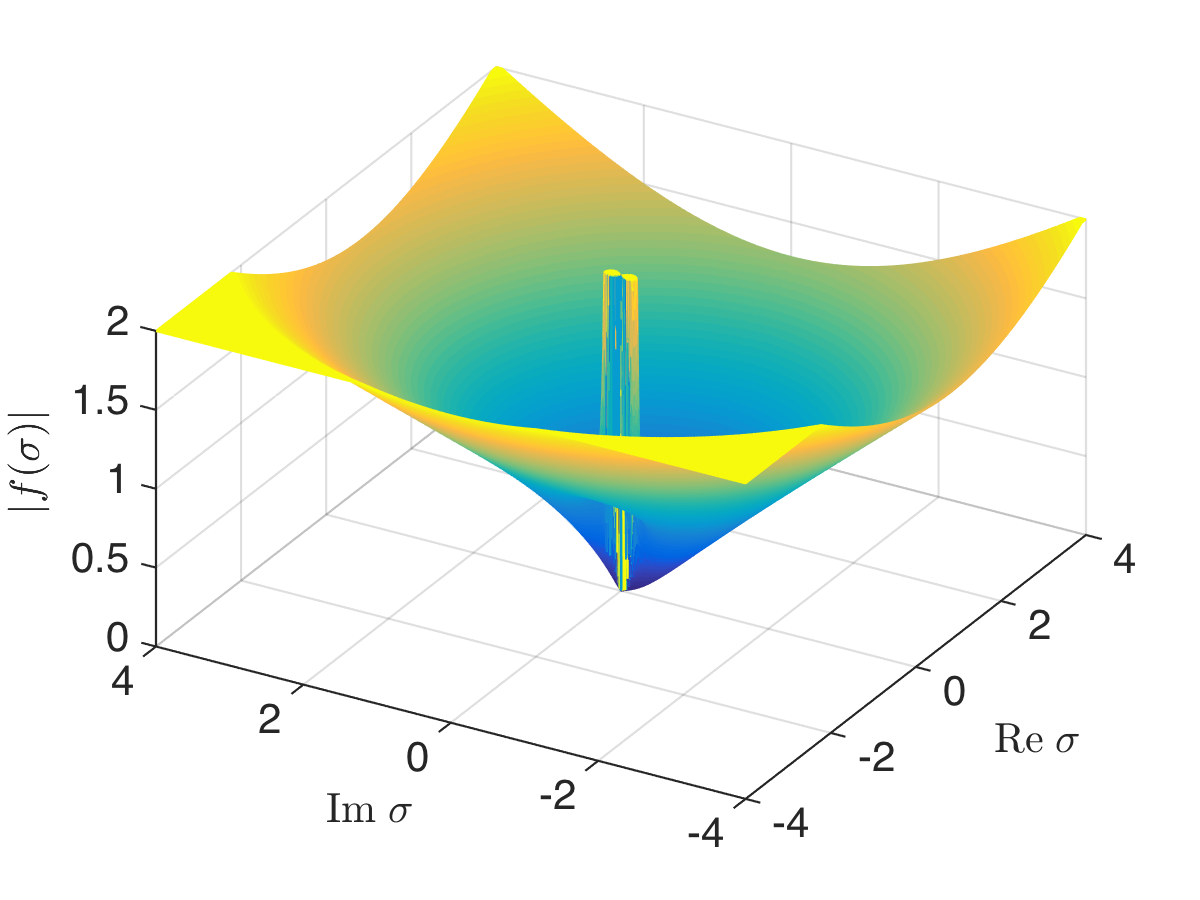}
 \caption{The modulus $|f(\sigma)|$ of the complex Black-Scholes price as a function of complex volatility $\sigma$. Note the gap along the real axis between the value at $0^-$ and $0^+$ as explained in remark~\ref{rmk:fractal_putcall}. The function has a similar behavior in a neighborhood of $\sigma = \infty$. The plot is truncated from above at 2 because in a neighbourhood of zero the surface goes infinitely many times to infinity.}\label{fig:modulusBlack-Scholes} % check
\end{figure}

A striking consequence of the previous result is that the implied volatility equation has now infinitely many solutions near zero.
\begin{figure}[!bt]
\centering
\subfloat[]{\includegraphics[width=.6\textwidth]{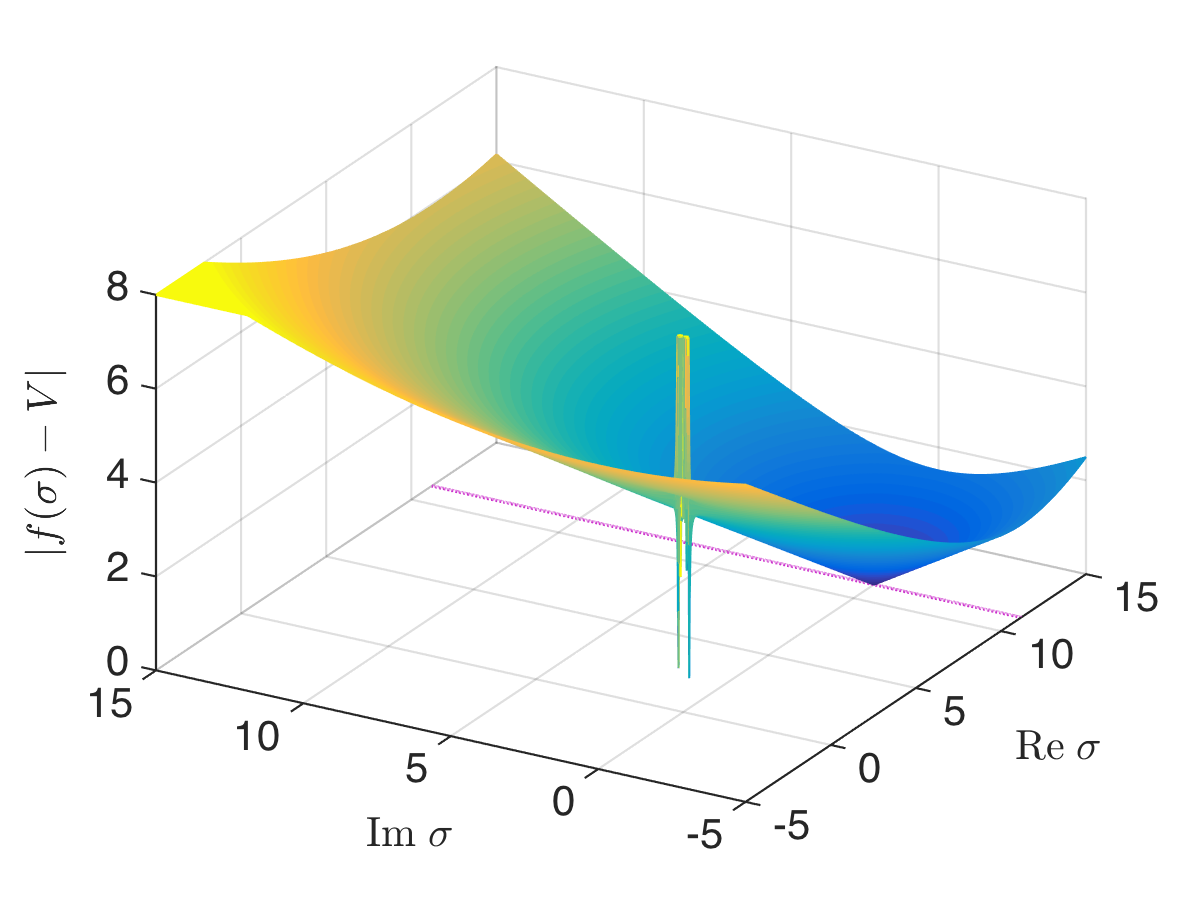}\label{fig:modulusBlack-Scholes-V}}\\
\subfloat[]{\includegraphics[width=.53\textwidth]{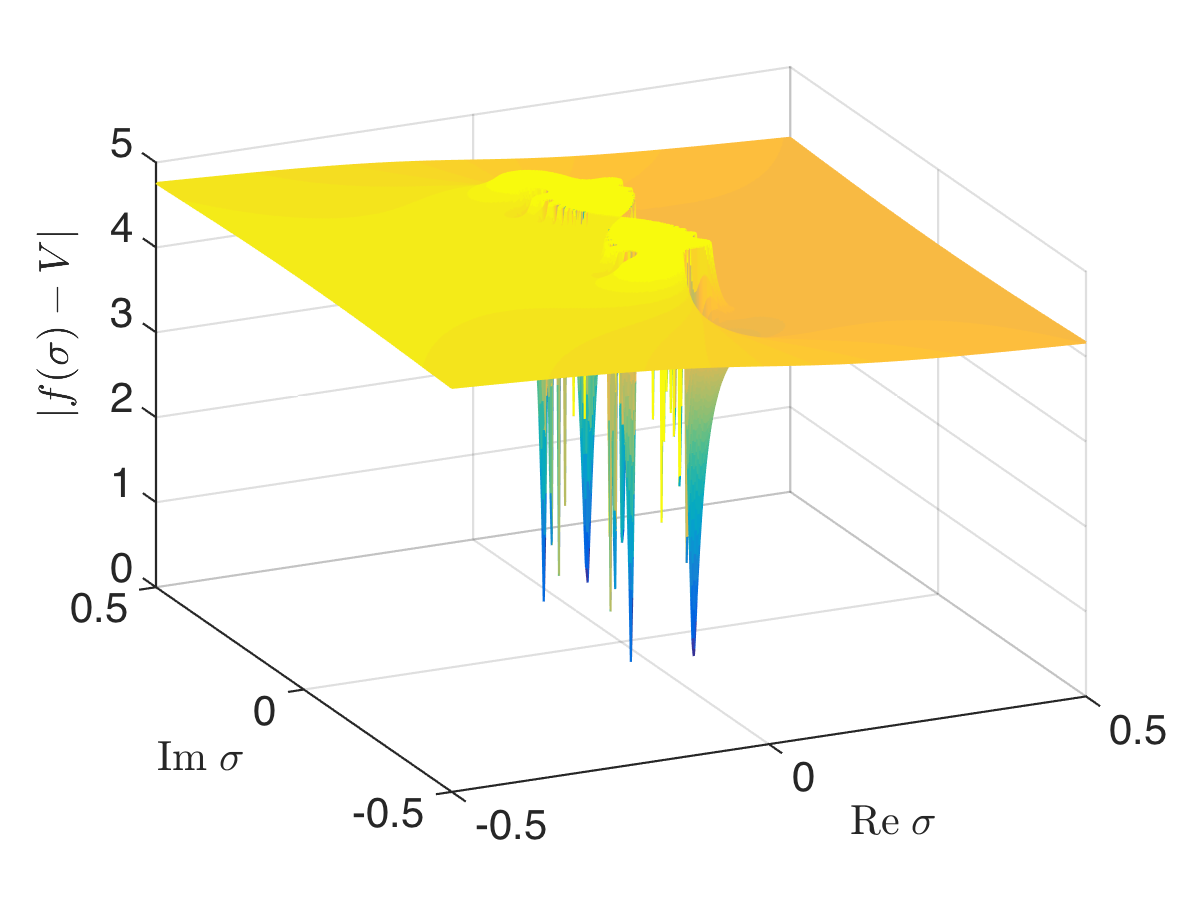}\label{fig:modulusBlack-Scholes-PV_zoomin3d}}%\hfill
\subfloat[]{\includegraphics[width=.53\textwidth]{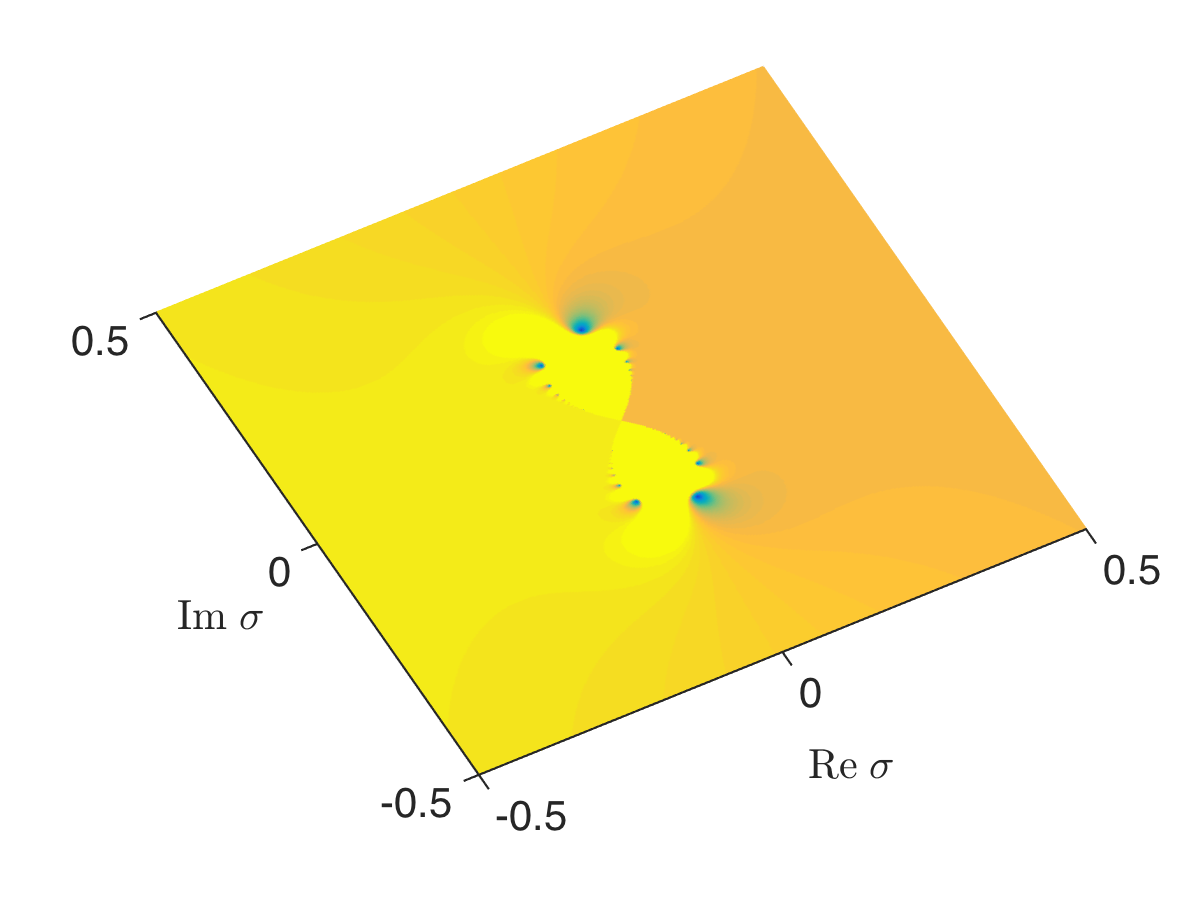}}
 \caption{The modulus of the calibration error $|f(\sigma)-V|$ as a function of complex volatility $\sigma$. (a) The minima of the surface correspond to the real implied volatility $\sigma=11.2$\,vol and to the infinitely many complex roots close to the origin; the plot is truncated from above at 8. (b) A blow-up of $|f(\sigma)-V|$ around the multiple complex roots close to the origin; the plot is truncated from above at 5. (c) Contour map of $|f(\sigma)-V|$. \label{fig:modulusBlack-Scholes-PV_zoomin}}
\end{figure}

\begin{corollary}\label{Casorati}
In each neighbourhood of 0, there are infinitely many complex values of the volatility $\sigma$ such that equation \eqref{Black-Scholespv} holds.
\end{corollary}

\begin{proof}
This is a consequence of the previous theorem and the Weierstrass-Casorati theorem \cite[Ch.~4, Thm.~9]{ahlfors1966}.
\end{proof}

Figure \ref{fig:modulusBlack-Scholes-PV_zoomin} displays the function $\vert f(\sigma)-V\vert $ for a fixed $V$ which according to corollary~\ref{Casorati} exhibits infinitely many complex zeros other than the real implied volatility.

In figures \ref{fig:modulusBlack-Scholes} -- \ref{fig:fractals} and figure \ref{fig:fractalJackel} we use market data of 28 November 2013 for a 1\,year at-the-money (ATM) USDJPY call option: the spot is $S=102.10$, the strike is $K=102.76$, and the ATM volatility is $\sigma=11.2$\,vol; the annual interest rate for JPY is $r=2.68\%$ and that for USD is $d=2.71\%$. %the 1Y forward points are -288 
The unit ``vol'' is a measure of volatility used habitually by practitioners, with $1 \text{ vol} = 1\%$.
We scale our complex plane in vols, so that the label 30 on the $x$ or $y$ axis means $30\%$ or $i30\%$. 

\FloatBarrier
\section{Implied volatility fractals}
\subsection{Newton-Raphson fractals}
Fractals arising from the Newton-Raphson algorithm have been studied in several 
articles \cite{curry1983, peitgen1984}. In fact the first fractals arose as an attempt to respond to the question by Cayley \cite{cayley1879}
on the loci of complex numbers converging to the multiple roots of a polynomial by the Newton-Raphson algorithm.

\begin{definition}
Given a function $g(x)$, for each positive integer $n$ we denote by $g^n$ the $n$-fold composition of $g$:
\begin{equation}
g^n(x)=\underbrace{g(g({\ldots}g(x){\ldots}))}_{n \text{ times}}.
\end{equation}
Then we can associate to each fixed point $x_*$ of $g$ its basin of attraction
\begin{equation}
B_g(x_*)=\left\{ x\in \mathbb{C}\vert \lim_{n\rightarrow \infty}g^n(x) =x_* \right\}.
\end{equation}
\end{definition}

If to solve a non-linear equation $y(x)=0$ we apply the Newton-Raphson iteration
\begin{equation}
x_{n+1}=x_n -\frac{y(x_n)}{y'(x_n)}
\end{equation}
and define
\begin{equation}
g(x) = x - \frac{y(x)}{y'(x)},
\end{equation}
then the initial values $x_0$ where the Newton-Raphson method converges to a given root $x_*$ are in the attraction basin $B_g(x_*)$. 

The Julia set $J_g$ is the boundary of the attraction basin $B_g(x_*)$ of a fixed point $x_*$ \cite{milnor2000}. It has been proved that when $g$ is a rational function the Julia set thus defined is independent of the fixed point and coincides with the closure of the repelling fixed points. Several illustrations of attraction basins coloured according to the convergence speed and also of Julia sets can be found in the book by Peitgen and Richter \cite{peitgen1986}.

\subsection{Fractals associated to implied volatility}
In this subsection, we present an empirical analysis of attraction basins and Julia sets for the Newton-Raphson method associated to equation \eqref{Black-Scholespv}. For simplification we use an equivalent equation given by J\"ackel \cite{jackel2006},
\begin{equation}\label{eq:Black-Scholescall_Jaeckelform}
b = h(\hat{\sigma})
\end{equation}
where $b := Ve^{rT}/\sqrt{FK}$,
\begin{equation}
h(\hat{\sigma}) := \theta e^{a/2}\Phi\left[\theta\left(\frac{a}{\hat{\sigma}}+\frac{\hat{\sigma}}{2}\right)\right]-\theta e^{-a/2}\Phi\left[\theta\left(\frac{a}{\hat{\sigma}}-\frac{\hat{\sigma}}{2}\right)\right],
\end{equation}
$a := \log(F/K)$, and $\hat{\sigma} := \sigma\sqrt{T}$.
The Newton-Raphson iteration to find the implied volatility from equation \eqref{eq:Black-Scholescall_Jaeckelform} is 
\begin{subequations}
\begin{align}
\hat{\sigma}_{n+1}
            &=\hat{\sigma}_{n} -\frac{h(\hat{\sigma}_n)-b}{h'(\hat{\sigma}_n)}\\
            &=\hat{\sigma}_{n} -\frac{\sqrt{2\pi}\left(h(\hat{\sigma}_n)-b\right)}{
            \exp{\left(-\frac{a^2}{2\hat{\sigma}_{n}^2} - \frac{\hat{\sigma}_{n}^2}{8}\right)}}.\label{eq:Black-Scholesiteration}
\end{align}
\end{subequations}
The termination criterion is 
\begin{equation}
\frac{\| \hat{\sigma}_{n+1}-\hat{\sigma}_{n} \|}{\| \hat{\sigma}_{n} \|}\leq \epsilon
\end{equation}
or $n = L$. Given an initial point $\hat{\sigma}_0$, a tolerance level $\epsilon$ and a maximum iteration number $L$, the Newton-Raphson iteration \eqref{eq:Black-Scholesiteration} terminates in three cases:
\begin{enumerate}
\item it converges to the real root;\label{stop2}
\item it converges to one of the many complex roots;\label{stop3}
\item it does not converge until the maximum iteration number is reached, or encounters other numerical problems.\label{stop1}
\end{enumerate}
We produced fractals by plotting each initial point $\hat{\sigma}_0$ in a different colour according to how the corresponding iterations terminate. Specifically, for the three termination cases above, we used the following colour scheme:
\begin{enumerate}
\item a shade of blue, linearly scaling from dark to light by the number of steps it takes to converge: the points in dark blue take fewer steps to converge than those in light blue;
\item a shade of red, linearly scaling from dark to light by the number of steps it takes to converge: the points in dark red take fewer steps to converge than those in light red;
\item black. 
\end{enumerate}

In this colouring scheme, the attraction basins are the regions that are not in black. In all figures, we used $1001 \times 1001$ initial points. Figure \ref{fig:fractals} shows an implied volatility fractal under different magnifications. Note that the enlargement of panel \ref{fig:fractal_1000} around the origin shown in panel~\ref{fig:fractal_04} is not observable in panel \ref{fig:fractal_1000} because of the limited resolution of the latter. 

\begin{figure}[!tb]
 \subfloat[]{\includegraphics[scale = 0.5]{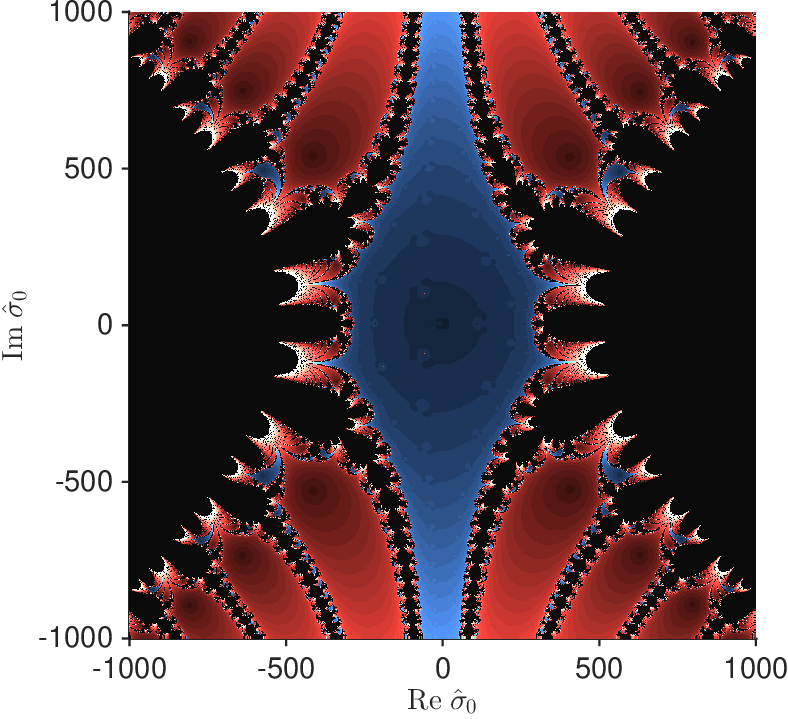}\label{fig:fractal_1000}}\hfill
 \subfloat[]{\includegraphics[scale = 0.5]{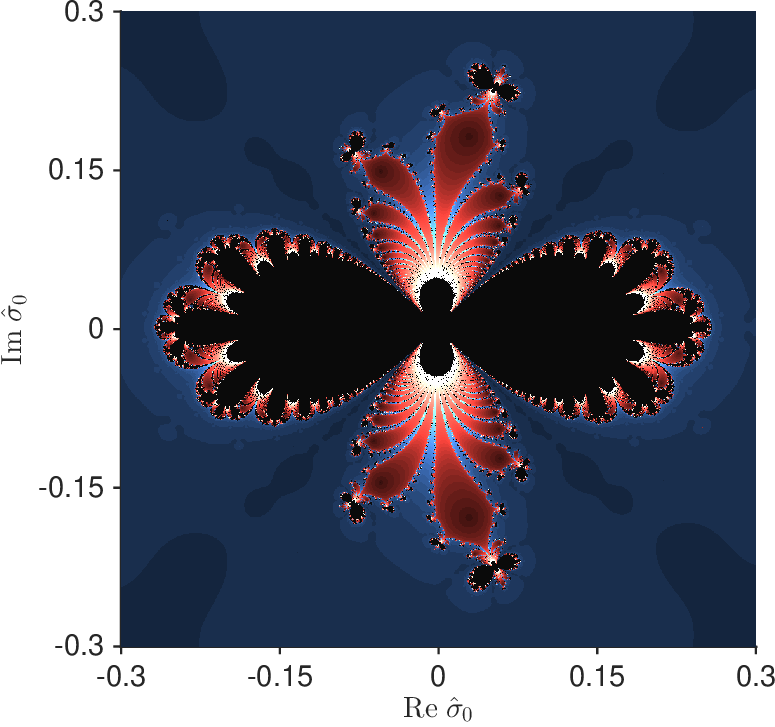}\label{fig:fractal_04}}\\
 \subfloat[]{\includegraphics[scale = 0.5]{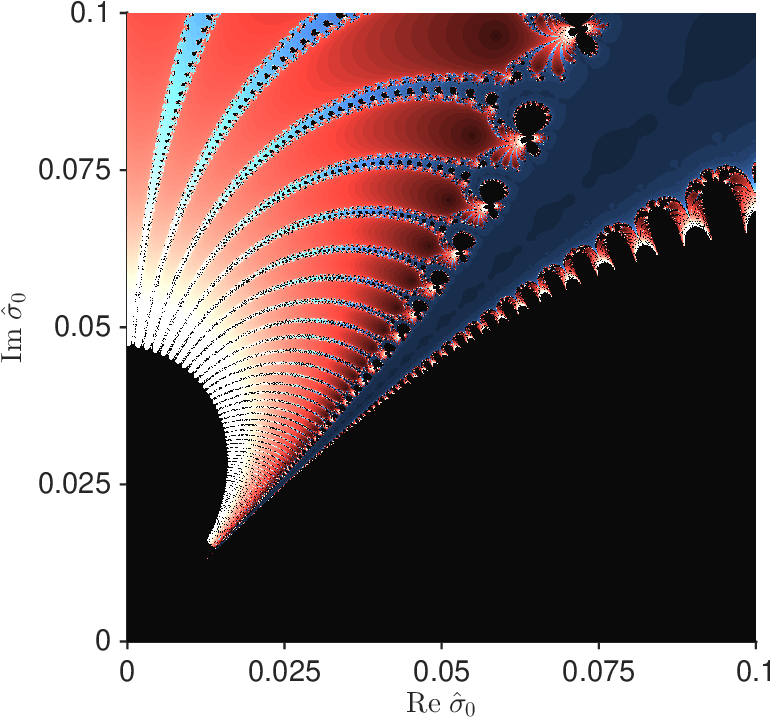}\label{fig:fractal_01}}
 \subfloat[]{\includegraphics[scale = 0.5]{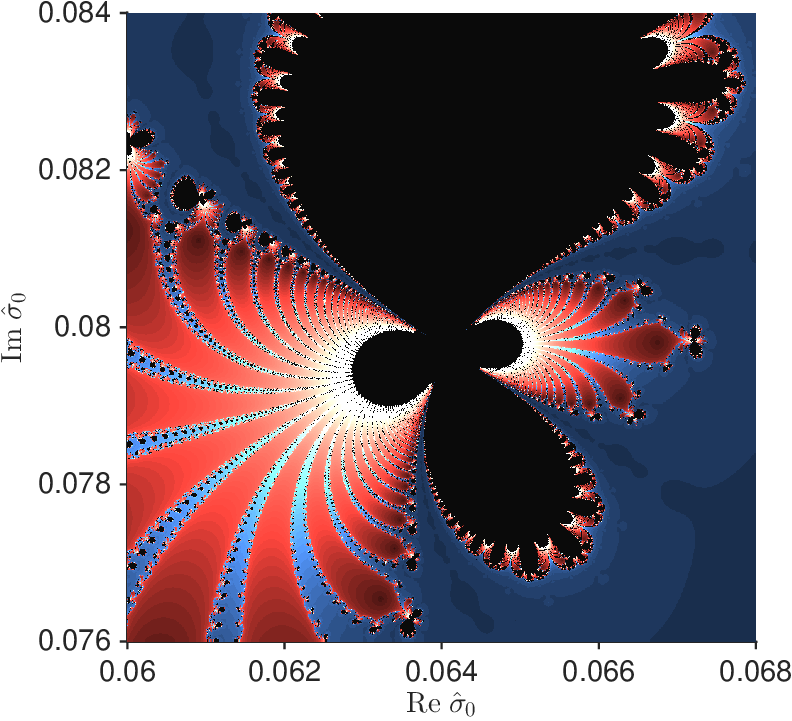}\label{fig:fractal_001}}
\caption{(a) Implied volatility fractal for an ATM option with $\epsilon = 10^{-8}$, $L=100$. The axis scale is in vols. (b) A zoom-in of (a) around the origin. (c) A zoom-in of (b) in the first quadrant. (d) A zoom-in of one petal in (c).}\label{fig:fractals}
\end{figure}

Figure \ref{fig:fractalsDeltas} shows the fractals for options with different values of $\Delta=\partial f/\partial S_0$, i.e.\ the rate of change of the option price with respect to the change in the spot price of the underlying. In the Black-Scholes model, $\Delta = \theta e^{-dt}\Phi\left[\theta\left(a/(\sigma\sqrt{T})+\sigma\sqrt{T}/2\right)\right]$. We use call and put options with $\Delta = 25\%$ and $\Delta = 10\%$ because they are often used by practitioners to depict the volatility smile (the ATM call option with $\Delta = 50\%$ is shown in figure \ref{fig:fractals}). The volatility smile is an important measurement indicating that the implied volatility changes with the strike. We show the fractals in the upper half of the complex plane because the fractals are symmetric with respect to the real axis, as can be seen in figure \ref{fig:fractals}.
Notice the difference in the number of attraction basins for different values of $\Delta$ and the similarity between call and put options with the same $\Delta$.

\begin{figure}[!tb]
\centering
\subfloat[Put, $\Delta=-10\%$, $K=76.74$, $\sigma=25$\,vol.]{\includegraphics[scale = 0.48]{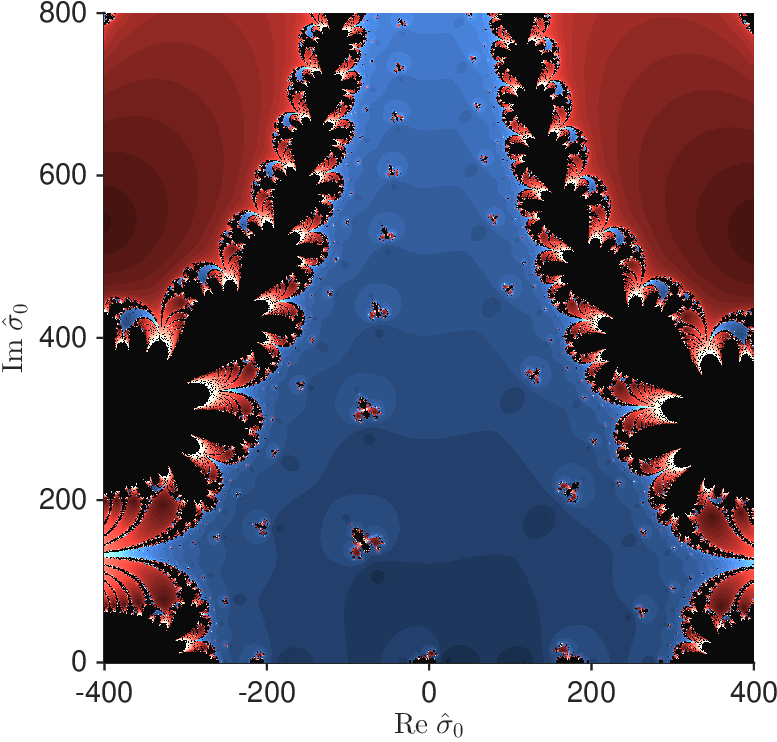}}\hfill
\subfloat[Put, $\Delta=-25\%$, $K=91.40$, $\sigma=20$\,vol.]{\includegraphics[scale = 0.48]{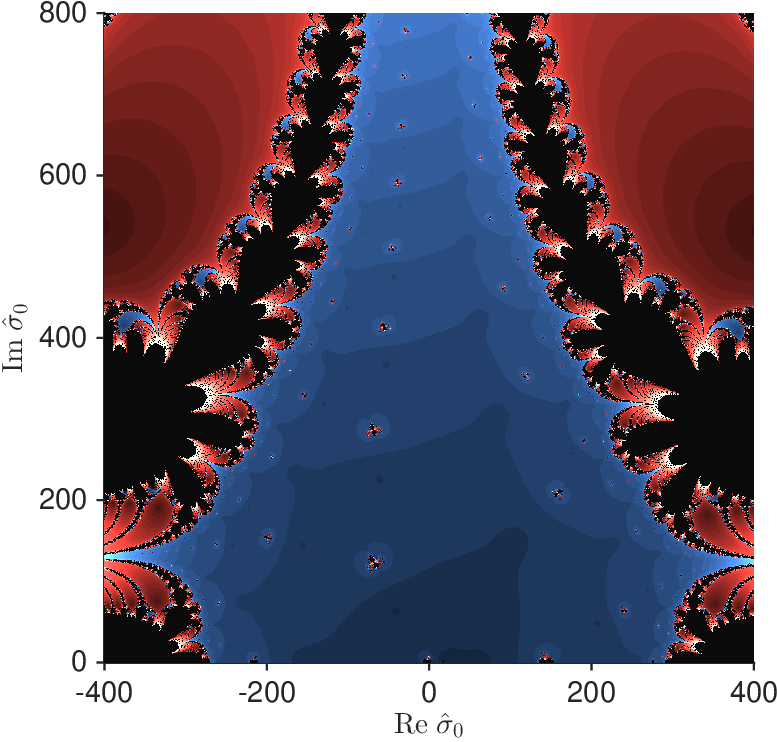}}\\
\subfloat[Call, $\Delta=10\%$, $K=126.50$, $\sigma=16$\,vol.]{\includegraphics[scale = 0.48]{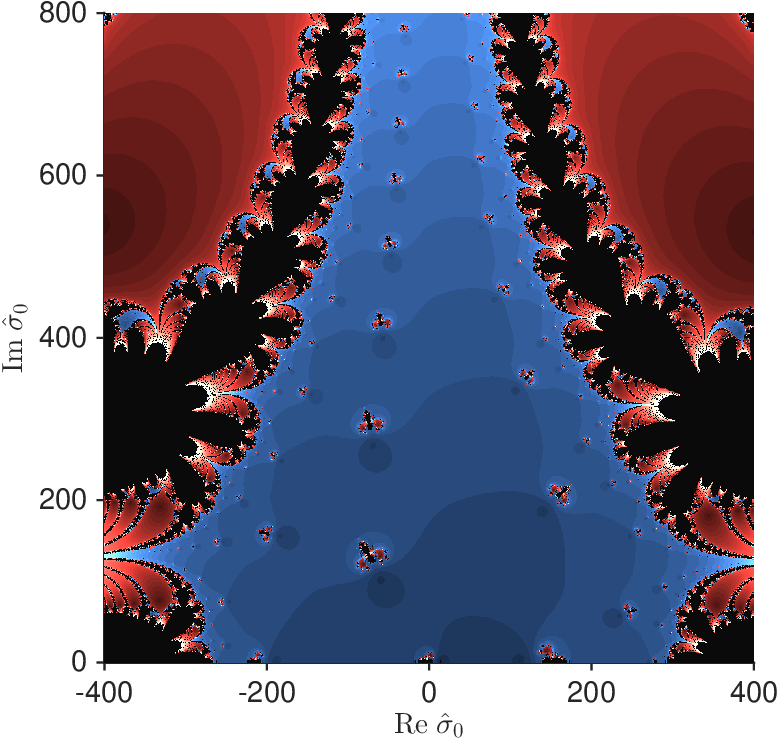}}\hfill
\subfloat[Call, $\Delta = 25\%$, $K=113.85$, $\sigma=15$\,vol.]{\includegraphics[scale = 0.48]{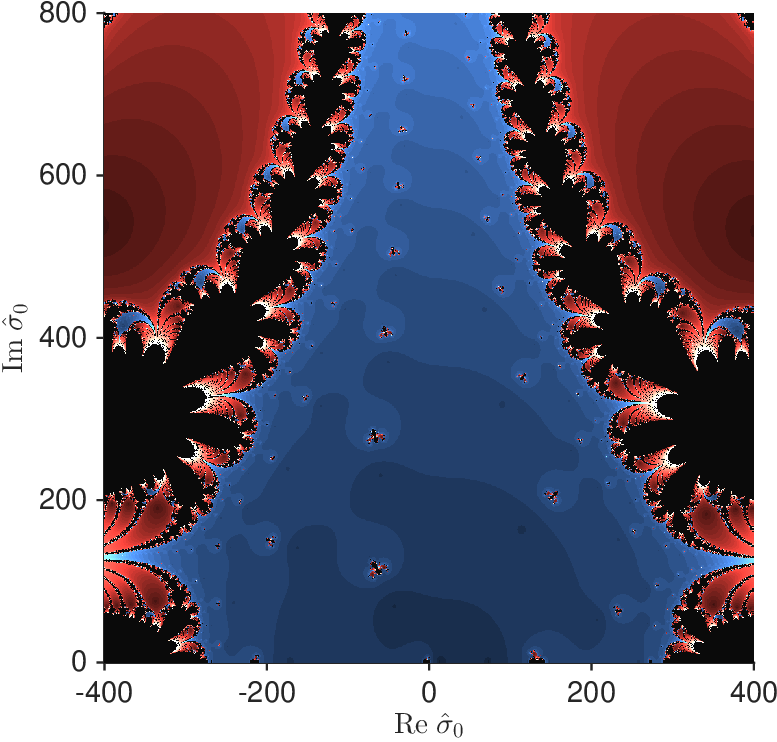}}
\caption{Implied volatility fractals for options with $\epsilon = 10^{-8}$, $L=100$ and various values of $\Delta$.}\label{fig:fractalsDeltas}
\end{figure}

Figure \ref{fig:fractalJackel} shows a fractal for an ATM option with the same parameters as figure \ref{fig:fractals} and J\"ackel's aforementioned modification of the equation to be solved \cite{jackel2006}. This modification subtracts from both sides of equation \eqref{eq:Black-Scholescall_Jaeckelform} the intrinsic value $\tau: = 2\theta H(\theta a)\sinh(a/2)$, where $H(\cdot)$ is the Heaviside function, and solves the equivalent form on a logarithmic scale
\begin{equation}\label{eq:Jaeckelmodification}
\log\frac{h(\hat{\sigma}) - \tau}{b - \tau} = 0.
\end{equation}
The corresponding Newton-Raphson iteration is
\begin{subequations}
\begin{align}
\hat{\sigma}_{n+1}
            &=\hat{\sigma}_n -\frac{\log\frac{h(\hat{\sigma}_n) - \tau}{b - \tau}}{\frac{1}{h(\hat{\sigma}_n) - \tau}h'(\hat{\sigma}_n)}\\
            &=\hat{\sigma}_n -\frac{\sqrt{2\pi}(h(\hat{\sigma}_n) - \tau)\log\frac{h(\hat{\sigma}_n) - \tau}{b - \tau}}{
            \exp{\left(-\frac{a^2}{2\hat{\sigma}_{n}^2} - \frac{\hat{\sigma}_{n}^2}{8}\right)}}.\label{eq:JackelIteration}
\end{align}
\end{subequations}
Notice that this modification largely reduces the red area (the initial points that lead to a convergence to complex roots) and the blue area (the initial points that lead to a convergence to the real root). The enlargement panel \ref{fig:JackelNearOrigin} shows the exquisite fractal near the origin.

\begin{figure}[!tb]
\centering
\subfloat[]{\includegraphics[scale = 0.48]{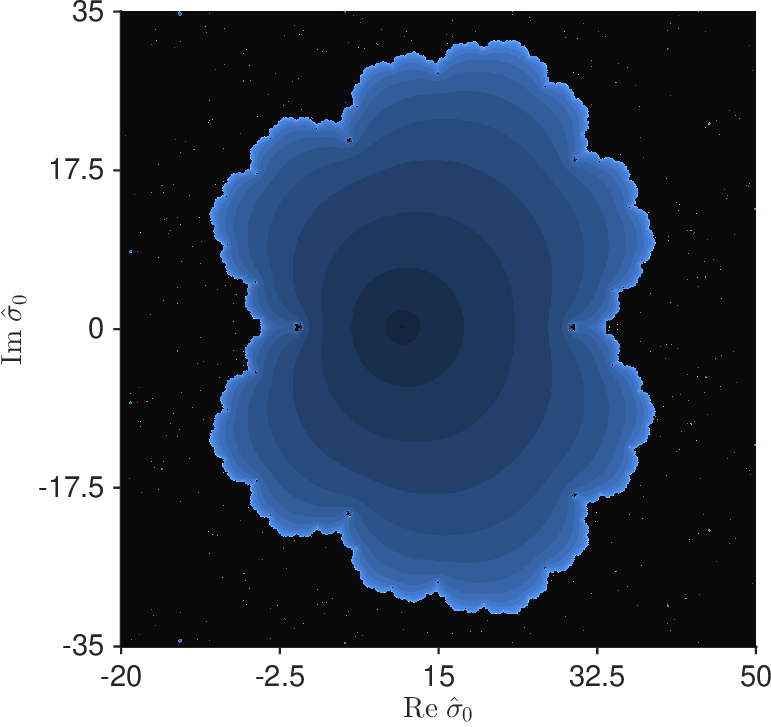}}\hfill
\subfloat[]{\includegraphics[scale = 0.48]{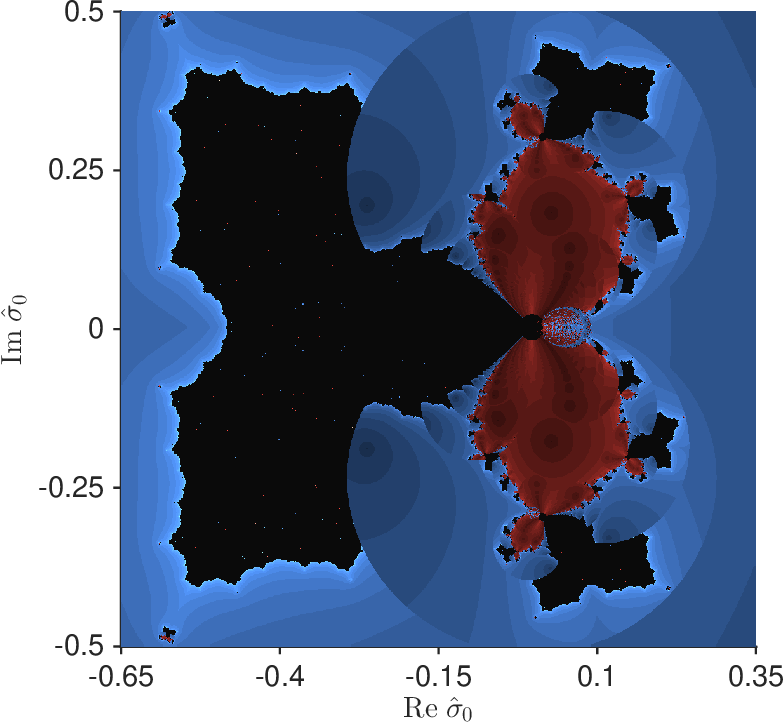}\label{fig:JackelNearOrigin}}
\caption{Implied volatility fractal for an ATM option with J\"ackel's modification, equation~\eqref{eq:Jaeckelmodification}.\label{fig:fractalJackel}}
\end{figure}

\FloatBarrier

\section{Conclusion}
We extended the Black-Scholes price as a function of the volatility to an analytic function on $\mathbb{C}_*=\mathbb{C}\setminus\{ 0\}$ and showed that the singularities at zero and infinity are essential. As a result, the objective function for finding the implied volatility has infinitely many complex solutions near zero. Following the professional practice, we adopt the Newton-Raphson method to resolve for the implied volatility. The chaotic nature of the calibration of the implied volatility is described in the complex plane by means of the associated fractal Julia sets. Fractals associated with the iterative process are shown for different moneyness values of this interesting problem very common in the financial industry. Among other things, these fractals visualise dramatically the effect of a modification suggested by J\"ackel to improve the stability and convergence of the search for the implied volatility.

\end{document}